%% file: expsos.tex
\documentclass{eptcs}

\newif\iflong
\longfalse

\title{Deriving Abstract Interpreters from Skeletal Semantics \iflong(Long Version)\fi}
\author{
Thomas Jensen
\institute{INRIA, Rennes}
\email{thomas.jensen@inria.fr}
\and
Vincent Rébiscoul
\institute{Université de Rennes, Rennes}
\email{vincent.rebiscoul@inria.fr}
\and
Alan Schmitt
\institute{INRIA, Rennes}
\email{alan.schmitt@inria.fr}
}

\newcommand{\titlerunning}{Deriving Abstract Interpreters from Skeletal Semantics}
\newcommand{\authorrunning}{T. Jensen, V. Rébiscoul \& A. Schmitt}

\hypersetup{
  bookmarksnumbered,
  pdftitle    = {\titlerunning},
  pdfauthor   = {\authorrunning},
  pdfsubject  = {EPTCS},               
}

\usepackage{amsmath}
\usepackage{amssymb}
\usepackage{stmaryrd}
\usepackage{mathtools}
\usepackage{mathpartir}
\usepackage[frozencache]{minted}
\setminted{fontsize=\small}
\renewcommand{\MintedPygmentize}{./pygmentize.py}
\usepackage{forest}
\usepackage{wrapfig}
\usepackage{breakurl}

\usepackage{subcaption}
\usepackage{ifthen}

\usepackage{multicol}

\usepackage{macros}

\newtheorem{theorem}{Theorem}
\newtheorem{definition}{Definition}
\newtheorem{lemma}{Lemma}

\newtheorem{proof}{Proof}

\begin{document}

\maketitle

\begin{abstract}
This paper describes a methodology for defining an executable abstract
interpreter from a formal description of the semantics of a
programming language. Our approach is based on Skeletal Semantics and
an abstract interpretation of its semantic meta-language.
The correctness of the derived abstract interpretation can be established by
compositionality provided that correctness properties of the core
language-specific constructs are established.
We illustrate the genericness of our method by defining a Value Analysis for a
small imperative language based on its skeletal semantics.
\end{abstract}

\section{Introduction}\label{sec:intro}

\input{intro}

\section{Skeletal Semantics}\label{sec:whileskelsem}

Skeletal Semantics offers a framework to mechanise semantics of programming
languages~\cite{bodin2019skeletal}. It uses a minimalist, functional, and
strongly typed semantic meta-language called
Skel~\cite{NoizetSchmitt2022}, whose syntax is presented in Figure~\ref{fig:bnf}.
%
The actual semantics of a language described in Skel is
expressed by providing a (meta-)interpretation of the Skel
language itself. In this paper, we will present two such interpretations: a
big-step (or concrete) semantics and an abstract interpretation.

We illustrate Skel through the definition of the skeletal semantics of a toy
imperative language called While. A Skeletal Semantics is a formal description
of a language and consists of \emph{declarations}. We start with some type
declarations (production \(r_{\tau}\) in Figure~\ref{fig:bnf}).
\begin{multicols}{3}
\inputminted[lastline=18]{sk}{while.sk}
\end{multicols}
\noindent{}
For While, there are four \emph{unspecified} types (identifiers,
literals, stores, integers) and
two \emph{specified} types (expressions and statements).
Unspecified types is an useful trait of Skel, their definitions are
unconstrained and they can be instantiated depending on the semantics of the
object language being defined. The specification of the integer type can be
different for a big-step semantics or for an abstract interpretation.
The \sktype{expr} and \sktype{stmt} types define expressions and
statements of While programs. 
An expression can be a
constant, a variable, an addition, a comparison, or a random integer. A
statement can be a skip (an instruction that does nothing), an assignment, a
sequence, a condition, or a loop. In addition to these declared types, one may
build arrow types and tuple types.

We now turn to Skel's \emph{term declarations} (production \(r_{t}\) of
Figure~\ref{fig:bnf}), which may also be unspecified or specified.
Unspecified terms are typically used for operations on values of unspecified
types. For our While language, they are as follows.
\begin{multicols}{2}
\inputminted[firstline=20, lastline=27]{sk}{while.sk}
\end{multicols}
The types for \skname{isZero} and \skname{isNotZero} may be surprising. These
partial functions act as filters when used in branches, as detailed below.

\begin{figure}
\begin{align*}
\textsc{Term} \quad t &\quad\Coloneqq\quad
x \mid
C \; t \mid
(t, .., t) \mid
\lambda p:\tau \rightarrow S \\
\textsc{Skeleton} \quad \skel & \quad\Coloneqq\quad
t \mid t_0~t_1 .. t_n \mid
\letin{p}{S}{S} \mid
\branch{\skel \skor.. \skor \skel}
\mid\\ & \quad\phantom{\Coloneqq}\quad \match{t}{p \rightarrow \skel .. p \rightarrow \skel}
\\
\textsc{Pattern} \quad p &\quad\Coloneqq\quad
x \mid
\_ \mid
C~p \mid
(p, .., p) \\
\textsc{{Type}} \quad \tau &\quad\Coloneqq\quad
b \mid
\tau \rightarrow \tau \mid
(\tau, .., \tau) \\
\textsc{{Term decl}} \quad r_t &\quad\Coloneqq\quad
\skkw{val}~x:\tau \mid
\skkw{val}~x:\tau = t
\\
\textsc{{Type decl}} \quad r_\tau &\quad\Coloneqq\quad
\skkw{type}~b
\mid \skkw{type}~b = \skinl{"|"}~ C_1 ~\tau_1 .. \skinl{"|"}~C_n ~\tau_n
\\
\textsc{{Skeletal Semantics}} \quad \mathcal{S} &\quad\Coloneqq\quad
(r_t | r_\tau)^*
\end{align*}
\caption{The Syntax of Skeletal Semantics}\label{fig:bnf}
\end{figure}

Specified terms, on the other hand, are signatures associated with a
\emph{term}. A term is either a skeletal variable, a constructor applied to a
term, a tuple, or an abstraction. The body of an abstraction is a skeleton,
described below. Consider the declaration of term \mintinline{sk}{eval_expr}.
\begin{multicols}{2}
\inputminted[firstline=29, lastline=44]{sk}{while.sk}
\end{multicols}
\noindent{}
The first line is syntactic sugar for
\begin{minted}[mathescape=t,escapeinside=||]{sk}
val eval_expr: (store, expr) -> int = |$\lambda$| (s, e): (store, expr) ->
\end{minted}
\noindent{}
where the remainder of the description is the body of the abstraction. This body
is a skeleton. A skeleton may be a term, an n-ary application, a let binding, a
branching (detailed below), or a match. Here the skeleton is a match,
distinguishing between the different expressions which may be evaluated. For a
constant expression, we call the unspecified term
\mintinline{sk}{litToInt} to convert the literal to an integer. For a
variable, we read its value in the store. For an addition, we sequence
the recursive evaluation of each subterm using a let binding, and we then apply
the unspecified \mintinline{sk}{add} term to perform the actual addition.
Note that specified term and type declarations are all mutually recursive.
The rest of the code does not use any additional feature.

We now turn to the second specified term declaration, to evaluate statements.
\begin{multicols}{2}
\inputminted[firstline=46, lastline=75]{sk}{while.sk}
\end{multicols}
\noindent{}
The code for the conditional and the loop illustrates the last feature of the
language, branching. Branches are introduced with the
\mintinline{sk}{branch} keyword and are separated with the
\mintinline{sk}{or} keyword. They correspond to a form of a
non-deterministic choice. Intuitively, in a big-step interpretation, any branch
that succeeds may be taken. Branches may fail if a pattern matching in a let
binding fails, or if the application of a term fails.
For instance, the instantiation of the term \skname{isNotZero} will not be defined on \(0\), making
the whole branch fail when given \(0\) as argument.
This is how we decide which branch to execute next for the
conditional, and whether to loop in the \mintinline{sk}{While} case.

\section{Big-step Semantics of Skel}
\label{sec:bigstep}

We give meaning to a Skeletal Semantics by providing \emph{interpretations} of
Skel. We first define the concrete, big-step semantics of Skel.
Let \skelsem{} be an arbitrary Skeletal Semantics. We write \(\funsof{\skelsem}\) for
the set of pairs \((\Gamma, \lambda{}p: \tau_1\rightarrow S_0)\) such that
\(\lambda{}p: \tau_1\rightarrow S_0\) appears in Skeletal Semantics \skelsem{}.
The typing environment \(\Gamma\) gives types to the free variables of \(\lambda{}p: \tau_1\rightarrow S_0\).
\iflong
\else{} Full formal details are available in~\cite{longversion}.\fi{}

\subsection{From Types to Concrete Values}\label{sec:typeint}

The definitions of the sets of semantic values are presented
on Figure~\ref{fig:buildval}. They are defined by induction on the
type. For each type \(\tau\), we write \(\itype{}{\tau}\) the set of values of type
\(\tau\). 

\begin{figure}\centering
  \begin{subfigure}{\textwidth}
    \begin{gather*}
      \itype{}{\tau_1\times..\times\tau_n} =
      \itype{}{\tau_1}\times..\times\itype{}{\tau_n}\\
      \itype{}{\tau_2} =
      \setof{C~v}{C:(\tau_1, \tau_2)\,\land\,v\in\itype{}{\tau_1}}\\
      \itype{}{\tau_1\to\tau_2} =
      \nclosures{\tau_1\to\tau_2}\cup\closures{\tau_1\to\tau_2}\\
    \end{gather*}
    \caption{Concrete values associated to each type}\label{fig:buildval}
  \end{subfigure}
  \hfill
  \begin{subfigure}{\textwidth}
    \centering
    \begin{gather*}
      \nclosures{\tau_1\to\tau_2} =
      \setof{\defclos{f}{n}}{
      \begin{gathered}
        \skkw{val}~f: \tau_1\to\tau_2 [= t]
        \in\skelsem \\
        \arity{f} = n
      \end{gathered}
      }\\
      \closures{\tau_1\to\tau_2} =
      \setof{
      (\Gamma, p, S, E)
      }{
      \begin{gathered}
        (\Gamma, \lambda{}p:\tau_1\to S)\in\funsof{\mathcal S}\\
        \envtype{E}{\Gamma} \\
        \gettype{S}{\extenv{\Gamma}{p}{\tau_1}}{\tau_2}
      \end{gathered}
      }\\
    \end{gather*}
    \caption{Named Closures and Anonymous Closures}
    \label{fig:closures}
  \end{subfigure}
  \caption{Definition of Concrete Values}
\end{figure}
A value with tuple type is a tuple of concrete values. A value of a specified type is a
constructor applied to a value. A value with arrow type is a function that can
either be a \emph{named closure} or an \emph{anonymous closure}.
The set of named closure \(\nclosures{\tau_1\to\tau_2}\) and the set of anonymous closures
\(\closures{\tau_1\to\tau_2}\) are defined on Figure~\ref{fig:closures}.
A named closure denotes a function that is specified in the Skeletal Semantics
\skelsem{}, it is a pair of the name of the function and its arity. An anonymous closure is
a tuple of a typing environment \(\Gamma\), a pattern \(p\) to bind the argument upon
application, a skeleton \(S\) which is the body of the function, and an
environment \(E\) captured at the creation of the closure. An environment is a
partial function mapping skeletal variable to concrete values.
It is said to be consistent with typing environment \(\Gamma\), written
\(\envtype{E}{\Gamma}\), if they have the same domain and if, for every
\(x \in \dom{\Gamma}\), we have \(E(x) \in \itypeabst{}{\Gamma(x)}\).

The unspecified types of a skeletal semantics must be instantiated to obtain an interpretation.
In the case of While, the unspecified types are \sktype{ident}, \sktype{lit},
\sktype{int}, and \sktype{store}. They are instantiated as follows.
\begin{align*}
  & \itype{}{\sktype{store}} = \setof{s}{s\in\mathcal{X}\hookrightarrow\mathbb{Z}}
  & \itype{}{\sktype{ident}} = \setof{x}{x\in\mathcal{X}}
  & \qquad \mbox{with } \mathcal{X} = \set{x, y, z, ..}\\
  & \itype{}{\sktype{lit}} = \setof{l}{l\in\mathbb{Z}}
  & \itype{}{\sktype{int}} = \setof{i}{i\in\mathbb{Z}}
\end{align*}
Identifiers are taken from a countable set \(\mathcal{X}\),
literals and integers are relative integers, and stores are partial maps from
identifiers to integers.

\subsection{Interpretation of Unspecified Terms}\label{sec:unspec-interp}

In the following, we write \(\notarr{\tau}\) when \(\tau\) is not an arrow type.
Take an unspecified term \(\skkw{val}~t: \tau_1\to..\to\tau_n\to\tau\) such that
\(\notarr{\tau}\), then an instantiation of \(t\), written \(\nonspec{t}\), is a
function such that
\(\nonspec{t}\in
(\itype{}{\tau_1}\times..\times\itype{}{\tau_n})\to\finparts{\itype{}{\tau}}\), where
\(\finparts{X}\) is the set of finite subsets of \(X\).
In particular, if \(\skkw{val}~t:\tau\) and \(\notarr{\tau}\), then
\(\nonspec{t}\subseteq \itype{}{\tau}\). Allowing the specification of a term to
be a function returning a set is useful to model non-determinism.

We instantiate the unspecified functions of our While language. The expression
\(\ternary{b}{e_1}{e_2}\) evaluates to \(e_1\) is the condition \(b\) is
true. Otherwise, it evaluates to \(e_2\).
\begin{align*}
  \nonspec{litToInt}(n) &= \set{n} &
  \nonspec{add}(n_1, n_2) &= \set{n_1+n_2}\\
  \nonspec{lt}(n_1, n_2) &= \ternary{n_1<n_2}{\set{1}}{\set{0}} &
  \nonspec{rand}(n_1, n_2) &= \setof{n}{n_1\le n \le n_2}\\
  \nonspec{isZero}(n) &= \ternary{n=0}{\set{()}}{\set{}} &
  \nonspec{isNotZero}(n) &= \ternary{n\neq0}{\set{()}}{\set{}}\\
  \nonspec{read}(x, s) &= \set{s(x)} &
  \nonspec{write}(x, s, n) &= \set{\extstate{s}{x}{n}}
\end{align*}
The \skname{rand} instantiation returns a set of values to capture the
non-determinism of the instruction. The \skname{isZero} function is defined only
on input 0, whereas \skname{isNotZero} is defined for all inputs except 0.

\subsection{Big-step Semantics}
\begin{figure}
  \begin{mathpar}
    \inferrule*[right=LetIn]
    {
      \evalskel{E, S_1}{v} \\ \addasn{}{\extenv{E}{p}{v}}{E'} \\
      \evalskel{E', S_2}{w}
    }
    {\evalskel{E, \sklet{} p = S_1 \skin{} S_2}{w}}
    \and
    \inferrule*[right=Branch]
    {\evalskel{E, S_i}{v}}
    {\evalskel{E, \left(S_1,.., S_n\right)}{v}}
  \end{mathpar}
  \caption{Examples of Rules of the Big-Step Semantics}\label{fig:natskel}
\end{figure}

\begin{figure}
  \begin{mathpar}
    \inferrule*[right=wild]
    { }
    {\addasn{}{\extenv{E}{\_}{v}}{E}}
    \and
    \inferrule*[right=var]
    { }
    {\addasn{}{\extenv{E}{x}{v}}{\set{x\mapsto v}E}}
    \and
    \inferrule*[right=const]
    {\addasn{}{\extenv{E}{p}{v}}{E'}}
    {\addasn{}{\extenv{E}{C\,p}{C\,v}}{E'}}
    \and
    \inferrule*[right=tuple]
    { \addasn{}{\extenv{E}{p_1}{v_1}}{E_2} \\
      .. \\
      \addasn{}{\extenv{E_n}{p_n}{v_n}}{E'}
    }
    {\addasn{}{\extenv{E}{(p_1,.., p_n)}{(v_1,.., v_n)}}{E'}}
  \end{mathpar}
  \caption{Rule of Extension of Environment using Pattern Matching}\label{fig:pattmatch}
\end{figure}

We briefly present the big-step semantics of Skel: \(\evalskel{E, S}{v}\) is
a relation from a skeletal environment \(E\), mapping skeletal variables to
values, and a skeleton \(S\) to a value \(v\). The relation is defined by
induction on \(S\) and is very similar to the natural semantics of \lcalc{} with
environment.
We focus on two of the most important rules on Figure~\ref{fig:natskel}.
The \textsc{LetIn} rule evaluates a let-binding by first evaluating \(S_1\),
next binding the result to the pattern in the current environment, then finally evaluating
\(S_2\) in the extended environment. The \textsc{Branch} rule describes how to
evaluate a branching:
any branch that successfully reduces to a value may be taken.
Finally, the pattern matching rules for environment extension are given in
Figure~\ref{fig:pattmatch}.
\iflong
\else{}The whole set of rules is given in~\cite{longversion}.\fi{}

The big-step semantics of a branching explains the types of some unspecified
terms seen earlier.
The partiality of the instantiations of \skinl{isZero} and \skinl{isNotZero}
functions are used in the semantics of While to prevent some branches to be
taken. They act as filters: if a branch does not have a derivation because its
filter is undefined on the input, the alternative is to take another branch.

\section{Big-step Semantics with Program Points}\label{sec:ppoints}


In this section, we introduce our first contribution, which is the integration
of the notion of \emph{program point} into the framework of Skeletal
Semantics. A program point maps-to a precise fragment of a given program.
They play an important role in semantics-based
program analysis, to indicate places where information about the
execution is collected. Program points are essential to abstract
interpretation as an abstract interpretation usually computes an abstraction of
the state of the execution of the analysed program for each program point.
Our formalisation of program points for Skeletal Semantics
is modular and works for the big-step semantics of Skel, but also for the
abstract interpretation of Skel, presented in Section~\ref{sec:abstint}.

In Skel, programs are values of an algebraic data type (ADT),
such as \skinl{stmt} or \skinl{expr} in the While example. For
instance, the skeletal term \skinl{Seq(Assign(x, Const 0), Assign(y, Const 1))} is a
While \textbf{program}
of type \skinl{stmt}. A program point is a path in the ADT of
the program, encoded as a list of integers (underlined to distinguish them from
natural numbers).

\begin{wrapfigure}{r}{0\textwidth}
  \begin{forest}
    [\skinl{Seq}
    [\skinl{Assign}, edge label={node[midway,left,font=\scriptsize]{\underline{0}}}
    [x, edge label={node[midway,left,font=\scriptsize]{\underline{0}}}]
    [\skinl{Const 0}, edge label={node[midway,right,font=\scriptsize]{\underline{1}}}]
    ]
    [\skinl{Assign}, edge label={node[midway,right,font=\scriptsize]{\underline{1}}}
    [y, edge label={node[midway,left,font=\scriptsize]{\underline{0}}}]
    [\skinl{Const 1}, edge label={node[midway,right,font=\scriptsize]{\underline{1}}}]
    ]
    ]
  \end{forest}
\end{wrapfigure}
For example, \(\epsilon\) is the empty path, it corresponds to the
whole program. The path
\underline{01} corresponds to \skinl{Const 0}.
The set of program points is thus \(\ppoint=\mathbb{N}^*\).

Let \(\prog\) be a term of an ADT and \(\pp\) a program point. We note
\(\prog@\pp\) the subterm of \(\prog\) at program point \(\pp\). Formally, it is
defined as follows.
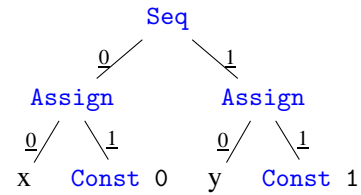
\begin{align*}
  & v@\underline{\epsilon} = v\qquad\qquad
  & \skinl{C}(v_0,.., v_{n-1})@(\underline{i} \pp) = v_i@\pp
  ~\mbox{ when } 0\leq i\leq n-1
\end{align*}

\subsection{Building Values with Program Points}\label{sec:ppointvals}

Our approach is to replace the values that correspond to programs with program points. These program points
correspond to a sub-program of a main program that is a parameter of the
interpretation. The values that ought to be replaced by program points should be
values representing fragments of the program being executed. Therefore, we
call \(\pptypes\) the set of \emph{program types}, i.e., types representing
programs. For instance, for the While language,
\(\pptypes=\set{\skinl{stmt}, \skinl{expr}}\). Moreover, the interpretation is
parametrised by a program \(\prog\), which is a value of a type
\(\tau\in\pptypes\). For the While language, that could be
\(\prog=\skinl{Seq(Assign(x, Const 0), Assign(y, Const 1))}\).
Therefore, the rules to build values are unchanged except for the values with
program types. Values with program types are defined by the following equation,
where \(\itype{}{\tau}\) is defined in Figure~\ref{fig:buildval}.
\begin{equation*}
  \tau\in\pptypes\implies
  \itypepp{\prog}{\tau} = \setof{\pp\in\ppoint}{\prog@\pp\in\itype{}{\tau}}
\end{equation*}
Therefore in our example, \(\epsilon\) is now a value of type \skinl{stmt}
denoting the value \(\prog\), and \underline{0} denotes the value
\skinl{Assign(x, Const 0)}.

For each unspecified term \(x\), we assume given an interpretation \nonspecpp{x},
which is identical to the concrete interpretation \nonspec{x} where program terms
are replaced by program points. The full definition of this interpretation can
be found in the long version of this paper~\cite{longversion}.

\subsection{Pattern-matching of Program Points}\label{sec:pattmatchpp}

Replacing some values with program points does not change the interpretation of
Skel, except when matching a program point with a pattern. Indeed, a program
point \(\pp\) corresponds to the sub-program \(\prog@\pp\) if it exists, and it
might be matched against a pattern \(\skinl{C}~p\).
To handle this case, the program point is \emph{unfolded}, meaning the
constructor at \(\prog@\pp\) is revealed, and the parameters of the
constructor are replaced with program points if their type is a program type.
To give an example, given \(\prog\) as before,
unfolding \(\epsilon\) gives \(\skinl{Seq}(\underline{0},
\underline{1})\): the constructor is revealed and the parameters are program
points because they both have type \(\skinl{stmt}\in\pptypes\). On the other
hand, unfolding \(\underline{00}\) directly returns \skinl{x}, as identifiers are
not program types in this example.
This unfolding mechanism is added to pattern matching via the following rule:
\begin{mathpar}
  \inferrule*[right=unfold]
  {\prog@\pp{} = \skinl{C}\,(v'_0, .., v'_{n-1}) \\ \skinl{C}:(\tau_0\times..\times\tau_{n-1}, \tau) \\
    v_j= \ite{\tau_j\in\pptypes}{\pp\underline{j}}{v_j'} \\
    \addasn{\pptypes, \prog}{\extenv{E}{p}{(v_0, .., v_{n-1})}}{E'}}
  {\addasn{\pptypes, \prog}{\extenv{E}{\skinl{C}\,p}{\pp}}{E'}}
\end{mathpar}
Note that the pattern matching is now parametrised with \(\pptypes\) and
\(\prog\). To perform the pattern-matching of \(\pp\) with \(\skinl{C}~p\), the
value \(\prog@\pp\) must have constructor \skinl{C} at the root. The
parameters that have a program type are replaced by their program
point and the pattern-matching is performed recursively.

\section{Abstract Interpretation of Skel}\label{sec:abstint}

We present our main contribution in this paper, which is the
definition of an abstract interpreter of Skel that is sound with respect
to the big-step semantics of Section~\ref{sec:ppoints}.
This abstract interpreter will serve as the foundations of a
methodology for building
abstract interpreters for object languages from their skeletal semantics. 
In this methodology, several ingredients must be provided to generate such an
abstract interpreter.
\begin{itemize}
\item An abstract instantiation of unspecified types to sets of abstract
  values, each of which comes with a concretisation function, a
  partial order, and an abstract union operator. Our framework
  automatically extends these definitions to all types.
\item A \emph{state of the abstract interpretation} (AI-state in
  the following) used to carry additional information during the
  abstract evaluation. For instance, in our While language, the
  AI-state records the current approximation of the store for each
  program point.
\item The user may give functions to update the AI-state at the
  start and end of a call to a specified function. This
  is typically used for evaluation functions, such as \skname{eval_stmt}, to record
  information right before and right after executing a sub-program.
\item 
  An instantiation of the unspecified terms, based on the 
  definition of 
  the abstract instantiation of types and the AI-state. 
\end{itemize}
Our framework provides an abstract meta-semantics of Skel that threads
the AI-state through the evaluation, including calls to unspecified
terms. As the foundational correctness property of the methodology, we
prove that if the abstract instantiation of types and terms provided
by the user satisfy some correctness criteria, then the whole abstract
interpreter that is generated is also correct.


\subsection{Abstract Values}\label{sec:abst-val}

\begin{figure}
  \centering
  \begin{gather*}
      \nclosures{\tau_1\to\tau_2}  =
      \setof{\defclos{f}{n}}{
        \begin{gathered}
          \skkw{val}~f: \tau_1\to\tau_2 [= t]
          \in\skelsem \\
          \arity{f} = n
        \end{gathered}
      }\\
      \closures{\tau_1\to\tau_2} =
      \setof{(\Gamma, p, S, E^\sharp)}{
        \begin{gathered}
          (\Gamma, \lambda{}p:\tau_1\to S)\in\funsof{\mathcal S}\\
          \envtype{E^\sharp}{\Gamma} \\
          \gettype{S}{\extenv{\Gamma}{p}{\tau_1}}{\tau_2}
        \end{gathered}
      }\\
      \itypeabst{}{\tau_1\times..\times\tau_n} = \finparts{\abstnonbot{\tau_{1}} \times .. \times \abstnonbot{\tau_{n}}}\\
      \itypeabst{}{\tau_{2}} = \{ C~v^\sharp \mid C:(\tau_{1},\tau_{2}) \land v^{\sharp} \in \abstnonbot{\tau_{1}} \} \cup \{\bot_{\tau_{2}}, \top_{\tau_{2}}\} \\
      \itypeabst{}{\tau_{1}\to\tau_{2}} = \parts{\nclosures{\tau_1\to\tau_2}} \cup \parts{\closures{\tau_1\to\tau_2}}
    \end{gather*}
    \caption{Abstract Values for Specified Types}\label{fig:buildabstval}
  \label{fig:abstvals}
\end{figure}
Abstract values are built similarly to concrete ones, based on their types, and
we write \(\itypeabst{}{\tau}\) for the set of abstract values of type \(\tau\). We
first define the sets of abstract (named) closures at the top of Figure~\ref{fig:buildabstval}.
Abstract named closures are identical to concrete named closures: they are pairs
of a name of a function defined in the skeletal semantics and its arity.
Abstract closures consist of a typing environment, a pattern, a skeleton, and an abstract
environment that is consistent with the typing environment. An abstract
environment \(E^{\sharp}\) is a mapping from skeletal variables to
abstract values. It is said to be
consistent with typing environment \(\Gamma\), written
\(\envtype{E^{\sharp}}{\Gamma}\), if they have the same domain and if, for every
\(x \in \dom{\Gamma}\), we have \(E^{\sharp}(x) \in \itypeabst{}{\Gamma(x)}\).

We assume that a partial order on \(\itypeabst{}{\tau}\) is provided for
each unspecified type \(\tau\), and that it includes
a smallest value, denoted by \(\bot_{\tau}\), and a largest value, denoted by
\(\top_{\tau}\). In the case of While, we instantiate
\skname{ident} with the flat lattice of \(\mathcal{X}\), \skname{lit} with the flat
lattice of integers, \skname{int} with closed intervals of \(\mathbb{Z}\cup\set{-\infty, +\infty}\), and \skname{store} with a partial
mapping from identifiers to non-empty intervals.

We define abstract values for specified types in Figure~\ref{fig:buildabstval}, writing \(\abstnonbot{\tau}\)
for \(\itypeabst{}{\tau} \setminus \{\bot_{\tau}\}\). Abstract tuples are finite sets of
tuples of (non-bottom) abstract values, with
\(\bot_{\tau_{1} \times .. \times \tau_{n}}\) being the empty set. We use sets to retain some precision in
the analysis.
Abstract values of an algebraic data type are simply
constructors applied to an abstract value of the correct type. Finally,
functional abstract values are sets of abstract closures of this type,
with \(\bot_{\tau_{1}\to\tau_{2}}\) being the empty set.


The AI-state \(\abststate\)
contains information collected throughout the abstract
interpretation. It is dependent on the analysis and the language, and therefore
must be provided, similarly to unspecified values.
Moreover, a partial order and a union must be given for abstract states.
In the case of while, the AI-state records information about the
abstract store before (\texttt{In}) and after (\texttt{Out}) every program
point. We write \texttt{Pos} for either \texttt{In} or \texttt{Out}. We then define \(\abststate\) as a mapping from
program points and a \texttt{Pos} to abstract while stores. We have
\(\abststate_{1} \aless \abststate_{2}\) if
\(\dom{\abststate_{1}} \subseteq \dom{\abststate_{2}}\) and for any
\((\pp,\mathtt{Pos}) \in \dom{\abststate_{1}}\), we have
\(\abststate_{1}(\pp,\mathtt{Pos}) \aless_{store} \abststate_{2}(\pp,\mathtt{Pos})\). We define
\(\abststate_{1} \join \abststate_{2}\) as the mapping from
\(\dom{\abststate_{1}} \cup \dom{\abststate_{2}}\) that relates \((\pp,\mathtt{Pos})\) to
\(\abststate_{1}(\pp,\mathtt{Pos}) \join_{int} \abststate_{2}(\pp,\mathtt{Pos})\).


\subsection{Operations on Abstract Values}\label{sec:abst-ops}

Each abstract domain has a partial order and an associated join
operator. In addition, a concretisation function that returns a set of
concrete values defines the meaning of each abstract value.  All of
these functions are indexed by types (or type environments when they
deal with environments). We assume they are provided for non-specified
types, and show in this section how to extend them to all types.

A concretisation function for type \(\tau\) maps an abstraction state and an
abstract value in \(\itypeabst{}{\tau}\) to \(\parts{\itype{}{\tau}}\), a set of
concrete values.  We also define a function of concretisation
\(\gamma_{\Gamma}\) which maps abstract skeletal environments to sets of
concrete skeletal environments.
\begin{align*}
  \gamma_{\tau_1\times..\times\tau_n}(\abststate, t^\sharp) =&
  \bigcup_{(v_1^\sharp,.., v_n^\sharp)\in t^\sharp}
  \gamma_{\tau_1}(\abststate, v^\sharp_1) \times..\times
  \gamma_{\tau_n}(\abststate, v^\sharp_n)\\
  \gamma_{\tau_2}(\abststate, C\,v^\sharp) =& \set{C\,v~|~ C:(\tau_{1}, \tau_2),\, v\in\gamma_{\tau_{2}}(\abststate, v^\sharp)}\\
  \gamma_{\tau_1\to\tau_2}(\abststate, F) =& \setof{\defclos{f}{n}}{\defclos{f}{n}\in F}
  \cup\setof{(\Gamma, p, S, E)}{(\Gamma,p, S, E^\sharp)\in F\,\land\, E\in\gamma_{\Gamma}(\abststate, E^\sharp)}\\
  \gamma_{\Gamma}(\abststate, E^\sharp) =& \setof{E}{
    \envtype{E}{\Gamma} \land \envtype{E^\sharp}{\Gamma} \land \forall x \in \dom{\Gamma}, E(x)\in\gamma_{\Gamma(x)}(\abststate, E^\sharp(x))
  }\\
  \gamma(\abststate, \bot_\tau) =&\ \emptyset\qquad
  \gamma(\abststate, \top_\tau) = \itypeabst{}{\tau}
\end{align*}
In the case of While, the concretisation function for \texttt{ident} and
\texttt{lit} are immediate as they are flat lattices. The concretisation
function for an interval \(i\) is the set of integers it contains:
\(\gamma_{\mathtt{int}}(i) = \setof{n}{n \in i}\). Finally, the concretisation
of an abstract store \(\sigma^{\sharp}\) is
\begin{equation*}
  \gamma_{\mathtt{store}}(\sigma^{\sharp}) = \setof{\sigma}{\dom{\sigma} = \dom{\sigma^{\sharp}} \land \forall x \in
    \dom{\sigma}, \sigma(x) \in \gamma_{\mathtt{int}}(\sigma^{\sharp}(x))}
\end{equation*}

To compare abstract values, we define partial orders that are relations, but we
call them functions as they can be viewed as boolean functions.
For every unspecified type \(\tau\), we assume a comparison
function \(\aless_{\tau}\) which is a partial order between abstract values.
It must satisfy the following property: for any value
\(v^{\sharp} \in \itypeabst{}{\tau}\), we have
\(\bot_{\tau} \aless_{\tau} v^{\sharp}\) and
\(v^{\sharp} \aless_{\tau} \top_{\tau} \).
For every other type, the comparison function is the smallest partial order that
satisfies the following equations.
\begin{gather*}
   C~v^\sharp\aless_{\tau_a} C~w^\sharp \iff
   v^\sharp\aless_\tau   w^\sharp \text{ with } C:(\tau, \tau_a)\\
   v^{\sharp} \aless_{\tau_1\times..\times\tau_n} w^{\sharp} \iff \forall (v_{1}^{\sharp},..,v_{n}^{\sharp}) \in v^{\sharp},\, \exists (w_{1}^{\sharp},..,w_{n}^{\sharp}) \in w^{\sharp} \text{
   such that }  \forall i \in [1..n],\,  v_{i}^{\sharp} \aless_{\tau_{i}}  w_{i}^{\sharp}
   \\
   F_1 \aless_{\tau_1\to\tau_2}  F_2 \iff
  \left\{
  \begin{aligned}
    \defclos{f}{n}\in F_1 &\implies \defclos{f}{n}\in F_2\\
    (\Gamma, p, S, E_1^\sharp)\in F_1 &\implies \exists (\Gamma, p, S, E_2^\sharp)\in F_2,\,  E_1^\sharp\aless_{\Gamma}  E_2^\sharp
  \end{aligned}
  \right.
  \\
     E_1^\sharp\aless_{\Gamma} E_2^{\sharp} \iff{}
    \envtype{E_1^\sharp}{\Gamma}\,\land\,\envtype{E_2^\sharp}{\Gamma}\,\land\, \forall{}
    x\in\dom{E_1^\sharp},\,  E_1^\sharp(x)\aless_{\Gamma(x)}  E_2^\sharp(x)
    \\
   v^\sharp\aless_{\tau}  \top_\tau\qquad
   \bot_\tau\aless_{\tau}  v^{\sharp}
\end{gather*}
Most rules are straightforward.
To compare two
functions, all named closures of the left function must be in the right
function. Moreover, for all closures in the left function, there must be a
closure in the right function with the same pattern and skeleton, but with a
bigger abstract environment. Abstract environments are compared using point-wise
lifting. For our While language, we have \(i_{1} \aless_{\mathtt{int}} i_{2}\)
if the interval \(i_{1}\) is included in \(i_{2}\), and
\(\sigma^{\sharp}_{1} \aless_{\mathtt{store}} \sigma^{\sharp}_{2}\) if for all
\(x\) in \(\dom{\sigma^{\sharp}_{1}}\), we have
\(\sigma^{\sharp}_{1}(x) \aless_{\mathtt{int}} \sigma^{\sharp}_{2}(x)\).

\begin{definition}\label{def:gamma-monotonic}
  A concretion function \(\gamma_{\tau}\) is \emph{monotonic} iff for any
  \(v_{1}^{\sharp} \aless_{\tau} v_{2}^{\sharp}\) and \(\abststate_{1} \aless \abststate_{2}\).
  we have
  \(\gamma_{\tau}(\abststate_{1}, v_{1}^{\sharp}) \subseteq \gamma_{\tau}(\abststate_{2},v_{2}^{\sharp})\).
\end{definition}
\begin{lemma}\label{lem:while-gamma-monotonic}
\(\gamma_{\mathtt{ident}}\), \(\gamma_{\mathtt{lit}}\), \(\gamma_{\mathtt{int}}\) and
\(\gamma_{\mathtt{store}}\) are monotonic.
\end{lemma}

For each type, an upper bound (or join) is defined.
For every non-specified type \(\tau\), we assume an upper
bound \(\join_{\tau}\). The definition of \(\join_{\texttt{ident}}\) and 
\(\join_{\texttt{lit}}\) have the usual definition for flat
lattices. \(\join_{\texttt{int}}\) is  the convex hull of two intervals and
\(\join_{\texttt{store}}\) is the usual point-wise lifting of the abstract union
of integers. Moreover, we note \(\nabla^\sharp_{\texttt{int}}\) the widening on intervals
(define below) and \(\nabla^\sharp_{\texttt{store}}\) the point-wise lifting of the
widening of intervals.
\begin{align*}
  & [n_1, n_2] \nabla^\sharp_{\texttt{int}} [m_1, m_2] =
    \left[
    \left\{
    \begin{array}{ll}
      n_1 & \mbox{ if } n_1\leq m_1\\
      -\infty & \mbox{ otherwise}
    \end{array}
    \right. ,
    \left\{
    \begin{array}{ll}
      n_2 & \mbox{ if } m_2\leq n_2\\
      +\infty & \mbox{ otherwise}
    \end{array}
    \right.
    \right]
\end{align*}
We extend it for every other type.
\begin{align*}
    & (C~v^\sharp)\join_{\tau_2} (C~w^\sharp) = C~(v^\sharp \join_{\tau_{1}} w^\sharp)
    \text{ with }C:(\tau_{1}, \tau_2) & 
    E_1^\sharp\join_{\Gamma} E_2^\sharp = \set{x\in\dom{\Gamma} \mapsto E_1^\sharp(x)\join_{\Gamma(x)} E_2^\sharp(x) }\\
    & (C~v^\sharp)\join_{\tau_2} (D~w^\sharp) = \top_{\tau_2}
    \text{ with }C:(\tau_{1}, \tau_2)\,\land\,D:(\tau'_{1}, \tau_2) &
    v^\sharp\join_\tau \top_\tau = \top_\tau\join_\tau v^\sharp = \top_\tau\\
    & v^{\sharp} \join_{\tau_1\times..\times\tau_n} w^{\sharp} = v^{\sharp} \cup w^{\sharp} &
    v^\sharp\join_\tau \bot_\tau = \bot_\tau\join_\tau v^\sharp = v^\sharp\\
    & F_1 \join_{\tau_1\to\tau_2} F_2 =  F_1 \cup F_2
  \end{align*} 
\noindent
Joining two algebraic values with the same constructor is joining their
parameters, and joining algebraic values with different constructors yields
top. The join of abstract tuples or abstract functions is their union.
Joining
abstract environments is done by point-wise lifting. For each type, top is an
absorbing element, and bottom is the neutral element.
\begin{lemma}
\(\sqsubseteq^\sharp_{\mathtt{ident}}\), \(\sqsubseteq^\sharp_{\mathtt{lit}}\), \(\sqsubseteq^\sharp_{\mathtt{int}}\),
\(\sqsubseteq^\sharp_{\mathtt{store}}\) are orders.
\(\sqcup^\sharp_{\mathtt{ident}}\), \(\sqcup^\sharp_{\mathtt{lit}}\), \(\sqcup^\sharp_{\mathtt{int}}\), 
\(\sqcup^\sharp_{\mathtt{store}}\) give an upper bound.
\end{lemma}

\begin{lemma}\label{lem:gamma-monotonic}
  If for all unspecified types \(\tau_u\), \(\gamma_{\tau_u}\) is monotonic,
  then for all \(\tau\), \(\gamma_\tau\) is also monotonic.
\end{lemma}

\begin{lemma}\label{lem:unspec-order-join}
  If for every unspecified type \(\tau_u\), \(\aless_{\tau_u}\) is an order and
  \(\join_{\tau_u}\) gives an upper bound, then for all \(\tau\), \(\aless_\tau\) is an
  order and \(\join_\tau\) gives an upper bound.
\end{lemma}

Finally, we give an abstract specification of unspecified terms. As an
illustration, here are a few specifications from our running example.
\begin{align*}
  \nonspecabst{litToInt}(n) &= [n, n] &
  \nonspecabst{add}([n_1, n_2], [m_1, m_2]) &= [n_1+m_1, n_2+m_2]\\
  \nonspecabst{read}(x, s^\sharp) &= s^\sharp(x) &
  \nonspecabst{write}(x, s, [n_1, n_2]) &= \extstate{s^\sharp}{x}{[n_1, n_2]}
\end{align*}

\begin{definition}\label{def:soundapprox}
  Let \(x\) be an unspecified term of type \(\tau\), such that \(\notarr{\tau}\). We
  say that \(\nonspecabst{x}\) is a \textbf{sound approximation} of
  \(\nonspecpp{x}\) if and only if:
  \[
    \forall\abststate,\, \nonspecpp{x}\subseteq\gamma(\abststate, \nonspecabst{x})
  \]
\end{definition}

\begin{definition}\label{def:soundapproxarrow}
  Let \(f\) be an unspecified term of type
  \(\tau_1\to..\to\tau_n\to\tau\) where \(\notarr{\tau}\). We
  say that  \(\nonspecabst{f}\) is a \textbf{sound approximation} of
  \(\nonspecpp{f}\) iff \(\forall v_i \in\itypepp{\prog}{\tau_i}\),
  \(\forall v^\sharp_i\in\itypeabst{}{\tau_i}\), and for all abstract state \abststate,
  if
  \begin{equation*}
    \left.
      \begin{aligned}
        v_i&\in\gamma_{\tau_i}(\abststate, v^\sharp_i) \\
        \nonspecabst{f}(\abststate, v^\sharp_1,.., v^\sharp_n) &= \abststate', w^\sharp
      \end{aligned}
    \right\}
    \implies
    \nonspecpp{f}(v_1,.., v_n)\subseteq\gamma_{\tau}(\abststate',w^\sharp)
  \end{equation*}
\end{definition}

\begin{lemma}\label{lem:while-soundapprox}
  The abstract instantiations of the unspecified terms for While are sound
  approximation of the concrete instantiations of the unspecified terms.
\end{lemma}

\subsection{Abstract Interpretation of Skel}\label{sec:abstint-rules}

The abstract interpretation of skeletons is given on
Figure~\ref{fig:abstintskel}.
It maintains a callstack of specified
function calls which is used to prevent infinite
computations by detecting identical nested calls.
A callstack is an ordered list of frames. The set of callstacks \(\Pi\) is defined
as:
\begin{mathpar}
  \inferrule*{ }{\empstack\in\cstacks{}}
  \and
  \inferrule*
  {
    \abststate{} \text{ an AI-state}\\
    \skkw{val}~f:\tau_1\to..\to\tau_n\to\tau\skinl{ = t}\in\skelsem\\
    \notarr{\tau} \\
    v_i\in\itypeabst{}{\tau_i}\\
    \cstack\in\cstacks\\
  }
  { \cons{(f, \abststate, [v_1, .., v_n])}{\cstack}\in\cstacks}
\end{mathpar}

\begin{figure}
    \begin{mathpar}
      \inferrule*[right=Var]{E^{\sharp} (x) = v^\sharp}{\evalterm{E^\sharp, x}{v^\sharp}}
      \and
      \inferrule*[right=TermClos]
      {\skkw{val}~f: \tau_1\to..\to\tau_n\to\tau[= t]\in\mathcal{S} \\ \notarr{\tau}}
      {\evalterm{E^\sharp, f}{\set{\defclos{f}{n}}}}
      \and
      \inferrule*[right=TermSpec]
      {\skkw{val}~x: \tau= t\in\mathcal{S} \\ \evalterm{\emptyset, t}{v^\sharp}\\ \notarr{\tau}}
      {\evalterm{E^\sharp, x}{v^\sharp}}
      \and
      \inferrule*[right=TermUnspec]
      {\skkw{val}~x: \tau\in\mathcal{S} \\ \notarr{\tau}}
      {\evalterm{E^\sharp, x}{\nonspecabst{x}}}
      \and
      \inferrule*[right=Const]
      {\evalterm{E^\sharp, t}{v^{\sharp}}}
      {\evalterm{E^\sharp, \operatorname{C}{t}}{\operatorname{C}{v^\sharp}}}
      \and
      \inferrule*[right=Tuple]
      {\evalterm{E^\sharp, t_1}{v_1^\sharp} \\ .. \\ \evalterm{E^\sharp, t_n}{v_n^\sharp}}
      {\evalterm{E^\sharp, (t_1,..,t_n)}{\set{(v_1^\sharp, .., v_n^\sharp)}}}
      \and
      \inferrule*[right=Clos]
      { }
      {\evalskelabst{\cstack, E^\sharp, \lambda{} p: \tau\cdot S}{\set{(p, S, E^\sharp)}}}
      \and
      \inferrule*[right=Branch]
      {
        \evalskelabst{\cstack, \abststate, E^\sharp, S_i}{v_i^\sharp,
          \abststate_i}\\
      }
      {
        \evalskelabst{\cstack, \abststate, E^\sharp, (S_1.. S_n)}
        {\join{} v_i^\sharp, \join\abststate_{i}}
      }
      \and
      \inferrule*[right=LetIn]
      { \evalskelabst{\cstack, \abststate, E^\sharp, S_1}{v^{\sharp}, \abststate'} \\
        \addasn{\pptypes, \prog}{\extenv{\set{E^\sharp}}{p}{v^\sharp}}
        {\set{E_1^\sharp,.., E_n^\sharp}} \\
        \evalskelabst{\cstack, \abststate', E^\sharp_i, S_2}{w^\sharp_i,
          \abststate_i}}
      {\evalskelabst{\cstack, \abststate, E^\sharp, \sklet{} p = S_1 \skin{} S_2}{\join{}
          w^\sharp_i, \join{} \abststate_i}}
      \and
      \inferrule*[right=App]
      {\evalterm{E^\sharp, t_i}{v_i^\sharp} \\ \evalapp{\cstack, \abststate,
          v_0^\sharp~v_1^\sharp.. v_n^\sharp}{v^{\sharp}, \abststate'}}
      {\evalskelabst{\cstack, \abststate, E^\sharp, t_0~t_1.. t_n}{v^{\sharp}, \abststate'}}
    \end{mathpar}
    \caption{Abstract Interpretation of Skeletons and Terms}\label{fig:abstintskel}
  \end{figure}

  \begin{figure}
    \begin{mathpar}
      \inferrule*[right=Spec]
      {
        \skkw{val}~f: \tau_1\to..\to\tau_n\to\tau{} = t\in\mathcal{S} \\
	\notarr{\tau} \\
        \evalterm{\emptyset, t}{v^\sharp}\\
        \updatein{f}{\abststate, [v_1^\sharp, .., v_n^\sharp]} = \abststate_1,
        [v_1'^\sharp, .., v_n'^\sharp] \\
        (f, \abststate_1, [v_1'^\sharp, .., v_n'^\sharp])\notin\cstack\\
        \evalapp{\cons{(f, [v_1'^\sharp, .., v_n'^\sharp])}{\cstack}, \abststate_1,
          v^\sharp~v_1'^\sharp.. v_n'^\sharp}{w^\sharp, \abststate_2}\\
        \updateout{f}{\abststate_2, [v_1'^\sharp, .., v_n'^\sharp], w^\sharp} =
        \abststate_3, w'^\sharp 
      }
      { \evalapp{\cstack, \abststate, (f, n)~v_1^\sharp.. v_n^\sharp}
        {w'^\sharp, \abststate_3}
      }
      \and
      \inferrule*[right=Spec-Loop]
      {
        \skkw{val}~f: \tau_1\to..\to\tau_n\to\tau{} = t\in\mathcal{S} \\
	\notarr{\tau} \\
        \evalterm{\emptyset, t}{v^\sharp}\\
        \updatein{f}{\abststate, [v_1^\sharp, .., v_n^\sharp]} = \abststate_1,
        [v_1'^\sharp, .., v_n'^\sharp] \\
        (f, \abststate_1, [v_1'^\sharp, .., v_n'^\sharp])\in\cstack\\
      }
      { \evalapp{\cstack, \abststate, (f, n)~v_1^\sharp.. v_n^\sharp}
        {\bot, \abststate_1}
      }
      \and
      \inferrule*[right=Unspec]
      {\skkw{val}~f: \tau_1\to..\to\tau_n\to\tau\in\mathcal{S} \\
	\notarr{\tau} \\
        \nonspecabst{f} (\abststate, v_1^\sharp,.., v_n^\sharp)
        = w^\sharp, \abststate'}
      {\evalapp{\cstack, \abststate, (f, n)~v_1^\sharp.. v_n^\sharp}{w^\sharp, \abststate'}}
    \end{mathpar}
    \caption{Abstract Interpretation: Application}\label{fig:abstintapp}
\end{figure}

The abstract interpretation of skeletons is similar to the big-step
interpretation: the evaluation of terms is almost unchanged except that
evaluating a closure or a tuple returns a singleton.
When evaluating a skeleton (branch, let-binding, or application), a state of the
abstract interpretation is carried through the computations.

In the
\textsc{Branch} rule, all branches are evaluated and joined instead of only one
branch being evaluated. Pattern matching now returns set of environments
rather than a single one (explained later).
As a consequence, the
\textsc{LetIn} rule may evaluate \(S_2\) in several abstract environments.
This flexibility in the control-flow of the abstract interpreter allows us to do
control flow analysis for \lcalc{}.
The \textsc{App} rule evaluates all terms and pass a list of values to the
application relation, defined in Figure~\ref{fig:abstintapp}.

Because the abstraction of a function is a set of closures and named closures,
the \textsc{App-Set} rule evaluates each one individually.
The \textsc{Base}
rule returns the remaining value when all arguments have been processed.
The \textsc{Clos} rule evaluates the body of the function
\(S\) in all abstract environments returned by the matching of the pattern
against the argument.
The \textsc{Spec} rule evaluates the call to a specified function and
maintains invariants. Invariants depend on the analysis and the AI-state,
therefore, language-dependent update functions can be provided to
to maintain invariants before and after a call.
The update functions must respect the following monotonicity
constraints in order to ensure soundness: 
\begin{definition}\label{def:updates-sound}
  The update functions are said to be monotonic if and only if:
  \begin{align*}
    \updatein{f}{\abststate, [v_1^\sharp, .., v_k^\sharp]}
     = \abststate', [v_1'^\sharp, .., v_k'^\sharp]
      &\implies
      \abststate\aless\abststate'\,\land\,
      (v_1^\sharp, .., v_k^\sharp) \aless (v_1'^\sharp, ..,
      v_k'^\sharp)\\
    \updateout{f}{\abststate, [v_1^\sharp, .., v_k^\sharp], v^\sharp}
    = \abststate', v'^\sharp  & \implies
      \abststate\aless\abststate'\,\land\,v^\sharp \aless v'^\sharp
  \end{align*}
\end{definition}

The update functions of While are defined as:
\begin{align*}
  \updatein{\texttt{eval\_stmt}}{\abststate, [s_i^\sharp, \pp]}
  & =
    \extstate{\abststate}{(\pp, \texttt{In})}{s^\sharp},
    [s^\sharp, \pp] & s^\sharp=s_i^\sharp\nabla^\sharp\abststate(\pp, \texttt{In})\\
  \updateout{\texttt{eval\_stmt}}{\abststate, [s_i^\sharp, \pp], s_o^\sharp}
  & =
    \extstate{\abststate}{(\pp, \texttt{Out})}{s^\sharp},
    [s^\sharp, \pp] & s^\sharp=s_o^\sharp\join \abststate(\pp, \texttt{Out})
\end{align*}
The update functions maintain the AI-state which holds an input and an output
abstract store for each program point.
\(\updatein{\texttt{eval\_stmt}}{\abststate, [s_i^\sharp, \pp]}\) updates the input
abstract store at program point \(\pp\) for a greater abstract store,
obtained by widening to ensure termination (discussed later), that
contains \(s_i^\sharp\), and the call to \(\texttt{eval\_stmt}\) is done with this
new abstract store.
\(\updateout{\texttt{eval\_stmt}}{\abststate, [s_i^\sharp, \pp], s_o^\sharp}\) makes a
similar change to the AI-state for the output store.
The update functions for \(\texttt{eval\_expr}\) (not presented here) do not
change the argument, the result or the AI-state.
\begin{lemma}\label{lem:upd-while-sound}
  The update functions for While previously defined are monotonic.
\end{lemma}

The extension of environments, or pattern matching, is presented \iflong
\else{} in~\cite{longversion} \fi{} and now returns a set \(\xi\) of
abstract environments as the abstraction of
tuples is a finite set of tuples of abstract values.
We thus return one abstract environment per tuple
of abstract values in our abstract tuple.

The termination of our analysis is not formally proven. Our intuition is that an
infinite derivation is necessarily caused by an infinity of calls to a specified
function. Given a program point, widening the stores in the input update
function for \texttt{eval\_st} should ensure that the input store converges and
that the \textsc{Spec-Loop} rule of the abstract interpretation ends the
computation, as we have reached a local fixpoint for the program point.

\subsection{Correctness of the Abstract Interpretation}\label{sec:abstint-correction}
Our methodology aims at defining mathematically correct abstract interpreters
from Skeletal Semantics. In this section, we present a theorem stating that the
abstract interpreter of Skel computes a correct approximation of the big-step
semantics of Skel.

We state the following theorem of correctness that states that the abstract
interpretation of Skel computes a sound approximation of the big-step
interpretation of Skel.
\begin{theorem}\label{thm:abstint-correction}
  Let \skelsem{} be a Skeletal Semantics with unspecified terms \(Te\) and
  unspecified types \(Ty\), and let \(E\) and \(E^\sharp\) be a concrete and abstract
  environment, respectively.
  Suppose
  \begin{itemize}
  \item \(\forall x\in Te\), 
    \(\nonspecabst{x}\) is a sound approximation of \(\nonspecpp{x}\).
  \item \(\forall\tau \in Ty\), \(\gamma_\tau\) is monotonic.
  \item \(\updatein{}{}\) and \(\updateout{}{}\) are monotonic.
  \end{itemize}
  Then:
  \begin{displaymath}
    \left.
    \begin{array}{c}
      E\in\gamma_{\Gamma}(\abststate_{0}, E^\sharp)\\
      \evalskel{E, S}{v}\\
      \evalskelabst{\empstack, \abststate_0, E^\sharp,S}{v^\sharp, \abststate}\\
    \end{array}\right\}
    \implies
    v\in\gamma(\abststate, v^\sharp)
  \end{displaymath}
\end{theorem}
Therefore, to prove the soundness of the analysis, it is sufficient to prove
that the abstract instantiation of terms are sound approximation of the
concrete ones, and that the update functions and concretisation functions are
monotonic.

Let \(\sigma_0\in\itypepp{}{\skname{store}}\) and \(\sigma^\sharp_0\in\itypeabst{}{\skname{store}}\) be
the concrete and abstract stores with empty domain.
Let \(E_0=\set{s\mapsto\sigma_0, t\mapsto\underline{\epsilon}}\) and \(E^\sharp_0=\set{s\mapsto\sigma_0^\sharp, t\mapsto\underline{\epsilon}}\) be
a concrete and an abstract Skel environments. We recall that \(\underline{\epsilon}\)
is the program point of the root of \(\prog\), the analysed program.
Let \(\abststate_0\) be the empty mapping from program points and flow
tags (\texttt{In} or \texttt{Out}) to abstract stores.
\begin{lemma}\label{lem:while-env-concr}
    \(\sigma_0\in\gamma_{store}(\abststate_0, \sigma_0^\sharp)\)
\end{lemma}

\begin{lemma}\label{lem:skel-env-concr-while}
  Let \(\Gamma=\set{s\mapsto\sktype{store}, t\mapsto\sktype{stmt}}\),
  \(E_0\in\gamma_\Gamma(\abststate, E_0^\sharp)\).
\end{lemma}

The
abstract interpreter computes an abstract store that is a correct
approximation of the concrete store returned by the big-step semantics.
\begin{theorem}\label{thm:while-correct}
    \begin{equation*}
      \left.
      \begin{gathered}
      \evalskelpp{E_0, \texttt{eval\_stmt}~(s, t)}{\sigma}\\
      \evalskelabst{\empstack, \abststate_0, E^\sharp_0, \texttt{eval\_stmt}~(s, t)}
      {\sigma^\sharp, \abststate}\\
    \end{gathered}
    \right\}
    \implies \sigma\in\gamma(\abststate, \sigma^\sharp)
  \end{equation*}
\end{theorem}

As an example, take
\(\prog \equiv \texttt{x := 0; while (x < 3) x := x + 1}\).
The concrete and abstract interpretations will find that
\begin{equation*}
  \begin{gathered}
    \evalskelpp{E_0, \texttt{eval\_stmt}~(s, t)}{\set{x\mapsto 3}}\\
    \evalskelabst{\empstack, \abststate_0, E^\sharp_0, \texttt{eval\_stmt}~(s, t)}
    {\set{x\mapsto[0, +\infty]}, \abststate}\\
  \end{gathered}
\end{equation*}
In accordance with Theorem~\ref{thm:while-correct}, we observe that
\(
  \set{x\mapsto3}\in\gamma(\abststate, \set{x\mapsto[0, +\infty]})
\)\\
The abstract interpreter returns an imprecise result. Currently, our method
fails to properly take into account the guards: the conditions of loops or
conditional branchings are not used to refine the abstract values. In the
previous While program, the guard of the loop is not used to get a precise
abstract store in or after the loop. The skeletal semantics of the While
language makes it unclear how to use the guards to modify the store, as it is
syntactically the same before and after the evaluation of the condition.

The precision of the analysis also depends on the skeletal semantics of the
language. An easy fix for our precision issue would be to modify the type of
\skinl{isZero} and \skinl{isNotZero} functions such that they have type
\((\texttt{store} \times \texttt{int}) \to \texttt{store}\).
The abstract instantiations of these functions could then be used to
refine the abstract stores.

\section{Related Work}\label{sec:related}

Our work is part of a large research effort to define sound analyses and build
correct abstract interpreters from semantic description of languages.
At its core, our approach is the Abstract
Interpretation~\cite{cousot1977abstract,cousot1979systematic} of a
semantic meta-language. Abstract Interpretation is a method designed by Cousot
and Cousot to define sound static analyses from a concrete
semantics. In his Marktoberdorf lectures~\cite{Cousot1998}, Cousot describes a systematic way
to derive an abstract interpretation of an imperative language from a
concrete semantics and 
mathematically proved sound. We chose to define the Abstract Interpretation of
Skel, as it is designed to mechanise semantics of languages. The benefit of
analysing a meta-language is that a large part of the work to define and prove
the correctness of the analysis is done once for every semantics mechanised with
Skel. However, it is often less precise than defining a language specific
abstract interpretation. Moreover, there have
been several papers describing methods to derive abstract interpretation from
different types of concrete
semantics~\cite{cousot1977abstract,schmidt1995natural,nielson1982denotational},
we chose to derive abstract interpreters from a big-step semantics of Skel.

Schmidt~\cite{schmidt1995natural} shows how to define an abstract
interpretation for \lcalc{} from a big-step semantics defined
co-inductively. The abstract interpretation of Skel and its correctness proof
follow the methods described in the paper.
However Skel has more complex constructs than \lcalc{}, especially
branches. Moreover, the big-step of Skel is defined inductively, thus reasoning
about non-terminating program is not possible. Also, to prove the correctness of
the abstract interpretation of Skel, we relate the big-step derivation tree to
the abstract derivation tree, similarly to Schmidt, but a key difference is that our
proof is inductive when Schmidt's proof is co-inductive.

Lim and Reps propose the TSL system~\cite{lim2013tsl}: a tool to define
machine-code instruction set and abstract interpretations. The specification of
an instruction set in TSL is compiled into a Common Intermediate Representation
(CIR). An abstract interpretation is defined on the CIR, therefore an abstract
interpreter is derivable from any instruction set description. However, the TSL
system is aimed at specifying and analysing machine code, and not languages in
general. Moreover, it is unclear how it would be possible to define analyses on
languages with more complex control-flow, like \lcalc{}.

In the paper on Skeletal semantics, Bodin \emph{et
  al.}~\cite{bodin2019skeletal} used skeletal semantics to relate concrete and
abstract interpretations in order to prove correctness. An important difference
between that work and the present is that their resulting abstract semantics is
not computable, whereas our abstract interpretation can be executed as an
analysis, as demonstrated by our implementation \cite{abstintgen}.
Moreover, our method computes an AI-state that collects information
throughout the interpretation and allows to use widening using the update
functions, rather than computing an Input/Output relation.

The idea of defining an abstract interpreter of a meta-language to define
analyses for languages has been explored, for example by Keidel, Poulsen and
Erdweg~\cite{keidel2018compositional}. They use
arrows~\cite{hughes2000generalising} as meta-language to describe
interpreters. The concrete and abstract interpreters share code
using the unified interface of arrows. By instantiating language-dependent parts
for the concrete interpretation and the abstract interpretation, they obtain two
interpreters that can be proven sound compositionally by proving that the
abstract language-dependent parts are sound approximation of the concrete
language-dependent parts, similarly to Skel. 
However, we chose to use a dedicated meta-language, Skel, as its
library~\cite{necrolib} makes defining interpreters for Skel convenient and one
objective is to use the NecroCoq tool~\cite{necrocoq} to generate mechanised
proofs that our derived abstract interpreters are correct.

\section{Conclusion}\label{sec:conclusion}

In this paper, we propose a methodology for  mechanically deriving 
correct abstract interpreters from mechanised semantics.
Our approach is based on Skeletal Semantics and its meta-language
Skel, used to write  a semantic
description of a language. It consists of two independent parts.
First, we define an abstract interpreter for Skel which is target language-agnostic and
is the core of all derived abstract interpreters from Skeletal Semantics.
The abstract interpreter of Skel is proved correct with respect to the
operational semantics of Skel. Second, for a given target language to analyse,
abstractions must be defined.
The abstract domains are defined by instantiating the
unspecified types and providing comparisons and abstract unions of
abstract values. The semantics
of the language-specific parts are defined by instantiating the unspecified
terms. By combining the abstract interpreter of Skel and  the abstractions of
the target language, we derive a working abstract interpreter specialised for
the target language, obtained by meta-interpretation of Skel.
We prove a theorem which states that the abstract interpreter of the
target language is correct if the abstract instantiation of the unspecified
terms are sound approximation of the concrete instantiation of the unspecified
terms. We illustrate our method to build abstract interpreters on two examples:
a value analysis for a small imperative language, and a CFA for \lcalc{}
\iflong
\else{}
(in the long version~\cite{longversion}).
\fi{}

The approach has been evaluated by an implementation of a tool~\cite{abstintgen}
to generate abstract interpreters from any skeletal semantics. It was tested on
While and \lcalc{} and resulted in executable, sound analyses
validating the feasibility of the approach. 


The current abstract interpreters that we obtain  have limitations to
their precision. Part of this imprecision stems from the fact that we
generate abstract
interpreters for any language based on an abstract interpreter for the
Skel meta-language skeletal semantics. An interesting feature of the
approach is that some precision can be gained in a generic fashion by
improving the underlying abstract interpretation of Skel. 
For example, our interval analysis for While does not refine the
abstract values when entering a part of the program guarded by a condition.
Take \skinl{If(Equal(x, 0), Skip, Assign(x, 0))}, evaluated in store where
\(\set{x\mapsto\top}\). Our abstract interpreter returns state \(\set{x\mapsto\top}\).
Indeed, the condition can be true or false thus both branches of the if
construct are evaluated but each one is computed in the store \(\set{x\mapsto\top}\)
because the condition is not used to refine the abstract values.
This issue can be addressed, e.g., by keeping a trace of the execution in order to know
if we are computing a statement guarded by a condition. Dealing with
this issue at the level of the meta-language analysis benefits all generated analyses.

\end{document}

%% file: intro.tex
The derivation of provably correct static analyses from a formal specification
of the semantics of a programming language is a long-standing
challenge. The recent advances in the mechanisation of semantics has
opened up novel perspectives for providing tool support for this
task, thereby enabling the scaling of this approach to larger
programming languages. 
This paper presents one such approach for mechanically constructing
semantics-based program analysers from a formal description of the
semantics of a programming language. We aim to provide
methodologies which not only can prove the correctness of program
abstractions but also lead to executable 
analysis techniques. Abstract Interpretation~\cite{cousot1977abstract}
has set out a methodology for defining 
an abstract semantics from an operational semantics and
for proving a correctness relation between abstract and concrete semantics
using Galois connections. The principle of abstract interpretation has
been applied to a variety of semantic frameworks, including small-step
and big-step (natural) operational semantics, and denotational
semantics.  
An example of this methodology is to build an abstract semantics from
a natural semantics~\cite{schmidt1995natural}.
Another example is Nielson's theory of abstract interpretation of
two-level semantics~\cite{nielson1989two} in which a semantic
meta-language is equipped with binding-time annotations so that types
and terms can be given a \emph{static} and \emph{dynamic}
interpretation, leading to different but (logically) related
interpretations.

In order for semantics-based program analysis to handle the
complexity of today's programming languages, it is necessary to
conceive a methodology that is built using some form of mechanised
semantics. Examples of this include 
Verasco~\cite{jourdan2016verasco}, a formally verified static analyser for the
C programming language. It uses abstract interpretation techniques to perform
value analyses, relational analyses\ldots{} Verasco is written in Coq and the
soundness of the analysis is guaranteed by a theorem: a program where the
analysis does not raise an alarm is free of errors.
Reasoning about program behaviours is possible as Verasco reuses the
formalisation of the C semantics in Coq that was written for
CompCert~\cite{leroy2009compcert}. CompCert is a proved semantic preserving
C compiler written in Coq.

Another example is the 
\K~\cite{rosu2010overview} framework for writing semantics using
rewriting rules. Rewriting rules make the formal
definition of a semantics both flexible and relatively simple to write, and
allows to mechanically derive objects from the semantics like an interpreter.
However, this mechanization can be complex:
\K-Java~\cite{bogdanas2015k} is a formalization of Java in \K, with close to four
hundred rewriting rules. It is unclear if it is possible to derive an
analysis from a mechanization in \K.

The key idea that we will pursue in this paper is
that an abstract interpreter for a semantic meta-language combined
with language-specific abstractions for a particular property yield a
correct-by-construction 
abstract interpreter for the specific language and property. 
We describe how to obtain a correct program analyser for
a programming language from its \emph{skeletal} semantics.
Skeletal Semantics~\cite{bodin2019skeletal} is a proposal for
machine-representable semantics of programming languages.

The skeletal semantics of a language \(\mathcal{L}\) is a partial description
of the semantics of \(\mathcal{L}\). Typically, a skeletal semantics will contain
definitions of the constructs of the language and functions of evaluation of
these constructs.
A skeletal semantics is written in the meta-language
Skel~\cite{NoizetSchmitt2022}, a minimalist functional language.
It is a \emph{meta language} to describe the semantics of \emph{object
  languages}. Skel has several semantics, called \emph{interpretations}, (small
step, big step~\cite{KhayamNoizetSchmitt2022}, abstract interpretation), giving
different semantics for the object languages.

\paragraph{Contributions}
\begin{itemize}
\item We propose new interpretations of the semantic meta-language
  Skel that integrates the notion of program point in a 
  systematic way.
\item We define an abstract interpretation for Skel. The abstract interpretation
  of Skel combined with language-specific abstractions define an analyzer for
  the object language.
\item We prove that the abstract interpretation of Skel is a sound approximation
  of the big-step interpretation of Skel, provided that some small
  language-dependent properties hold.
\item We implement a program which, given a Skeletal Semantics, generates an
  executable abstract interpreter, and we test it on toy languages. We 
  define a basic value analyzer for a small imperative language.
  A Control Flow Analysis for a \(\lambda\)-calculus is also presented in
  \iflong\else{}the long version of this paper~\cite{longversion}\fi{}.
\end{itemize}
